\documentclass[a4paper,11pt,twoside]{article}
\usepackage{graphicx,epsfig}
\usepackage{epstopdf}
\usepackage{comment}
\usepackage{mathrsfs}
\usepackage{amsfonts}
\usepackage{amsmath}
\usepackage{amssymb}
\usepackage{appendix}

\begin{document}
\newtheorem{theorem}{Theorem}
\newtheorem{proof}{Proof}
\newtheorem{acknowledgement}{Acknowledgement}

\title{An estimation of the logarithmic timescale in an ergodic dynamics}
\author{
        Ignacio S. Gomez \\
                Instituto de F\'{i}sica La Plata, Universidad Nacional de La Plata,\\
        CONICET, Facultad de Ciencias Exactas\\
        Calle 115 y 49, 1900 La Plata, \underline{Argentina}\\
        stratoignacio@hotmail.com
}

\maketitle
\begin{abstract}
An estimation of the logarithmic timescale in quantum systems having an ergodic dynamics in the semiclassical limit of quasiclassical large parameters, is presented. The estimation is based on the existence of finite generators for ergodic measure--preserving transformations having a finite Kolmogorov--Sinai (KS) entropy and on the time rescaling property of the KS--entropy. The results are in agreement with the obtained in the literature but with a simpler mathematics and within the context of the ergodic theory.
\end{abstract}

\section{Introduction}
%%%%%%%%%%%%%%%%%%%%%%%%%%%%%%%%%%%%%%%%%%%%%%%%%%%%%%%%%%%%%%%%%%%%%%%%%%%%%%%

It is well known that the manifestation of the dynamical aspects of quantum chaos is possible only within its characteristic timescales, the Heisenberg time in the regular case and the logarithmic timescale in the chaotic case, where typical phenomena with a semiclassical description are possible such as relaxation, exponential sensitivity, etc \cite{Ber89,Cas95,Gut90,Haa01,Sto99}.
The logarithmic timescale determines the time interval where the wave--packet motion is as random as the classical trajectory by spreading over the whole phase space \cite{Hel89}. It should be noted that some authors consider the logarithmic timescale as a satisfactory resolution of the apparent contradiction between the Correspondence Principle and the quantum to classical transition in chaotic phenomena \cite{Cas95}.
%the quantum transient (finite-time) and the evidence that time and classical limits do not commute.
Concerning the chaotic dynamics in quantum systems, the Kolmogorov--Sinai entropy (KS--entropy) \cite{Lic92,Tab79,Wal82} has proven to be one of the most used indicators. The main reason is that the behavior of chaotic systems of continuous spectrum can be modeled from discretized
models such that the KS--entropies of the continuous systems and of the discrete ones coincide for a
certain time range. Taking into account the graininess of the quantum phase space due the Indetermination Principle, non commutative quantum extensions of the KS--entropy can be found \cite{Ben04,Ben05,Cri93,Cri94,Fal03}. Thus, the issue of graininess is intimately related
to the quantum chaos timescales \cite{Eng97,Gomez14,Ike93,Lan07}.
% and must necessarily be compatible with the restriction on these timescales.

To complete this picture, it should be mentioned that a relevant property in dynamical systems is the ergodicity, i.e. when the subsets of phase space have a correlation decay such that any two subsets are statistically independent ``in time average" for large times. This property, that is assumed as an hypothesis in thermodynamics and in ensembles theory \cite{Hua87,Pat72}, is at the basis of the statistical mechanics by allowing the approach to the equilibrium by means of densities that are uniformly distributed in phase space. In this sense, in previous works \cite{Cas09,Gom15} quantum extensions of the ergodic property were studied, from which characterizations of the chaotic behavior of the Casati--Prosen model \cite{Cas05} and of the phase transitions of the kicked rotator were obtained \cite{Gom14}.

The main goal of this paper is to exploit the graininess of the quantum phase space and the properties of the KS--entropy in an ergodic dynamics in order to get an estimation of the logarithmic timescale in the semiclassical limit.
The paper is organized as follows. In Section 2 we give the preliminaries and an estimation of the KS--entropy in an ergodic dynamics is presented. In Section 3 we show that this estimation can be extended for the classical analogue of a quantum system in the semiclassical limit. From this estimation and a time rescaled KS--entropy of the classical analogue, we obtain the logarithmic timescale. Section 4 is devoted to a discussion of the results and its physical relevance.
Finally, in Section 5 we draw some conclusions, and future research directions are outlined.

%\paragraph{Outline}
%The remainder of this article is organized as follows.
%Section~\ref{previous work} gives account of previous work.
%Our new and exciting results are described in Section~\ref{results}.
%Finally, Section~\ref{conclusions} gives the conclusions.

%%%%%%%%%%%%%%%%%%%%%%%%%%%%%%%%%%%%%%%%%%%%%%%%%%%%%%%%%%%%%%%%%%%%%%%%%%%%%%%
\section{Preliminaries}
%%%%%%%%%%%%%%%%%%%%%%%%%%%%%%%%%%%%%%%%%%%%%%%%%%%%%%%%%%%%%%%%%%%%%%%%%%%%%%%
The definitions, concepts and theorems given in this Section are extracted from the Ref. \cite{Wal82}.
\subsection{Kolmogorov Sinai entropy}
We recall the definition of the KS--entropy within the standard framework of measure theory.
Consider a dynamical system given by $(\Gamma, \Sigma, \mu, \{T_t\}_{t\in J})$, where $\Gamma$ is the phase space, $\Sigma$ is a $\sigma$-algebra, $\mu:\Sigma \rightarrow [0,1]$ is a normalized measure and $\{T_t\}_{t\in J}$ is a semigroup of measure--preserving transformations. For instance, $T_t$ could be the classical Liouville transformation or the corresponding classical transformation associated to the quantum Schr\"{o}dinger transformation. $J$ is usually $\mathbb{R}$ for continuous dynamical systems and $\mathbb{Z}$ for discrete ones.

Let us divide the phase space $\Gamma$ in a partition $Q$ of $m$ small cells $A_{i}$ of measure $\mu (A_{i})$. The entropy of $Q$ is defined as
\begin{equation}\label{entropy partition}
H(Q)=-\sum_{i=1}^{m}\mu(A_{i})\log\mu(A_{i}).
\end{equation}
Now, given two partitions $Q_1$ and $Q_2$ we can obtain the partition $Q_1\vee Q_2$ which is $\{a_i\cup b_j: a_i\in Q_1, b_j\in Q_2\}$, i.e. $Q_1\vee Q_2$ is a refinement of $Q_1$ and $Q_2$.
In particular, from $Q$ we can obtain the partition $H(\vee_{j=0}^{n}T^{-j}Q)$ being $T^{-j}$ the inverse of $T_{j}$ (i.e. $T^{-j}=T_{j}^{-1}$) and $T^{-j}Q=\{T^{-j}a:a\in Q\}$.
From this, the KS--entropy $h_{KS}$ of the dynamical system is defined as
\begin{equation}\label{KS-entropy}
h_{KS}=\sup_{Q}\{\lim_{n\rightarrow\infty}\frac{1}{n}H(\vee_{j=0}^{n}T^{-j}Q)\}
\end{equation}
where the supreme is taken over all measurable initial partitions $Q$ of $\Gamma$. In addition, the Brudno theorem expresses that the KS--entropy is the average unpredictability of information of all possible trajectories in the phase space. Furthermore, Pesin theorem relates the KS--entropy with the exponential instability of motion given by the Lyapunov exponents. Thus, from the Pesin theorem it follows that $h_{KS}>0$ is a sufficient condition for chaotic motion \cite{Lic92,Tab79}.
\subsection{Time rescaled KS--entropy}
By taking $(\Gamma, \Sigma, \mu, \{T_t\}_{t\in J})$ as the classical analogue of a quantum system and considering the timescale $\tau$ within the quantum and classical descriptions coincide, the definition \eqref{KS-entropy} can be expressed as
\begin{equation}\label{KS-entropy1}
h_{KS}=\sup_{Q}\{\lim_{n\tau\rightarrow\infty}\frac{1}{n\tau}H(\vee_{j=0}^{n\tau}T^{-j}Q)\}
\end{equation}
Now since $T^{-j\tau}=(T_{\tau})^{-j}$ one can recast \eqref{KS-entropy1} as
\begin{equation}\label{KS-entropy2}
h_{KS}=\frac{1}{\tau}\sup_{Q}\{\lim_{n\rightarrow\infty}\frac{1}{n}H(\vee_{j=0}^{n}(T_{\tau})^{-j}Q)\} \nonumber
\end{equation}
Finally, from this equation one can express $h_{KS}$ as
\begin{equation}\label{KS-entropy3}
h_{KS}=\frac{1}{\tau}h_{KS}^{(\tau)}   \ \ \ , \ \ \ h_{KS}^{(\tau)}=\sup_{Q}\{\lim_{n\rightarrow\infty}\frac{1}{n}H(\vee_{j=0}^{n}(T_{\tau})^{-j}Q)\}
\end{equation}
The main role of the time rescaled KS--entropy $h_{KS}^{(\tau)}$ is that allows to introduce the timescale $\tau$ as a parameter. This concept will be an important ingredient for obtaining the logarithmic timescale.

\subsection{Ergodicity}
In dynamical systems theory, the correlation decay of ergodic systems is one of the most important properties for the validity of the statistical description of the dynamics because different regions of phase space become statistical independent ``in time average" when they are enough separated in time. More precisely, if one has a dynamical system $(\Gamma,\mu,\Sigma,\{T_t\})$ that is ergodic then the correlations between two arbitrary sets $A,B\subseteq \Gamma$ that are sufficiently separated in time satisfy
\begin{eqnarray}\label{ergodic}
lim_{T\rightarrow\infty}\frac{1}{T}\int_{0}^{T}C(T_tA,B)dt=0 \ \ \ \ , \ \ \ \ \textrm{for all} \ A,B\subseteq\Gamma \nonumber
\end{eqnarray}
where $C(T_tA \cap B)=\mu(T_tA \cap B)-\mu(A)\mu(B)$ is the correlation between $T_tA$ and $B$ with $T_tA$ the set $A$ at time $t$. This equation expresses the so called \emph{ergodicity property} which guarantees the equality between the time average and the space average of any function along the trajectories of the dynamical system. Ergodicity property is satisfied by all the chaotic systems, like chaotic billiards, chaotic maps, including systems described by ensemble theory. Furthermore, the calculation of the KS--entropy is intimately related with the ergodicity property, as we shall see.

%expresses that the time average is the same for almost all initial: statistically speaking, the system that evolves for a long time "forgets" its initial state

\subsection{Some methods for calculating the Kolmogorov--Sinai entropy}

In order to calculate the Kolmogorov--Sinai entropy $h_{KS}$ the concept of generator is of fundamental importance. A numerable partition $\widetilde{Q}=\{a_1,a_2,\ldots,a_i,\ldots\}$ of $\Gamma$ is called a \emph{generator} of $\Gamma$ for an invertible measure--preserving $T$ if
\begin{eqnarray}\label{generator}
\vee_{n=-\infty}^{\infty}T^{n}\widetilde{Q}=\Sigma
\end{eqnarray}
This equation expresses that the entire $\sigma$--algebra $\Sigma$ can be generated by means of numerable intersections of the form \\
\noindent $\ldots \cap T^{-2}a_{k_{-2}}\cap T^{-1}a_{k_{-1}} \cap a_{k_{0}} \cap T^{1}a_{k_{1}} \cap T^{2}a_{k_{2}} \cap \ldots$ where $a_{k_j}\in \widetilde{Q}$ for all $j\in\mathbb{Z}$. It can be proved that if $\widetilde{Q}$ is a generator and $H(\widetilde{Q})<\infty$ then
\begin{eqnarray}\label{generator2}
h_{KS}=\lim_{n\rightarrow\infty}\frac{1}{n}H(\vee_{j=0}^{n}T^{-j}\widetilde{Q}) \nonumber
\end{eqnarray}
which reduces the problem of taking the supreme in the formula of the $h_{KS}$ to find a generator $\widetilde{Q}$. In practice, still having found a generator, the calculation of $H(\vee_{j=0}^{n}T^{-j}\widetilde{Q})$ turns out a difficult task due to the large number of subsets of the partition $\vee_{j=0}^{n}T^{-j}\widetilde{Q}$ as soon as $n$ increases. However, a good estimation of the $h_{KS}$ can be made by means of the existence of finite generators. This is the content of the following theorem.
\begin{theorem}\label{estimation KS}(Estimation of the KS--entropy by means of finite generators)

\noindent If $(\Gamma, \Sigma, \mu, \{T_t\}_{t\in J})$ and $T=T_{1}$ is an ergodic invertible measure--preserving transformation with $h_{KS}<\infty$ then $T$ has a finite generator $\widetilde{Q}$
\begin{eqnarray}\label{generator3}
\widetilde{Q}=\{a_1,a_2,\ldots,a_n\}
\end{eqnarray}
such that
\begin{eqnarray}\label{generator4}
e^{h_{KS}}\leq n \leq e^{h_{KS}}+1
\end{eqnarray}
\end{theorem}

\section{Logarithmic timescale in an ergodic dynamics in the semiclassical limit}
With the help of the Theorem \ref{estimation KS} and taking into account the graininess of the quantum phase space, one can obtain a semiclassical version of the Theorem \ref{estimation KS} from which the logarithmic timescale can be deduced straightforwardly. We begin by describing the natural graininess of the quantum phase space.

\subsection{A quantum generator in the quantum phase space}
\begin{figure}[th] \label{fig1}
\begin{center}
\includegraphics[width=8cm]{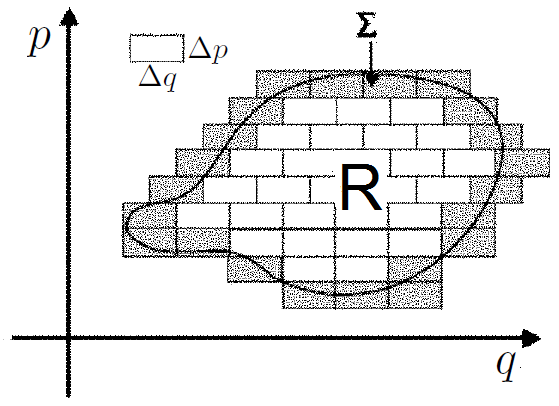}
\end{center}
\caption{Bounded motion and graininess in quantum phase space. In the semiclassical limit $q\gg1$ the region $R$ that the classical analogue occupies has a volume that is approximately the sum of the volumes of the rigid boxes $\Delta q\Delta p$ contained in $R$. The region $\Sigma$ corresponding to the rigid boxes that intersect the frontier of $R$ can be neglected in the limit $q\gg1$.}
\end{figure}
In quantum mechanics, the Uncertainty Principle leads to a granulated phase space composed by rigid boxes of minimal size $\Delta q\Delta p\geq\hbar^{D}$ where $2D$ is the dimension of the phase space. This is the so called \emph{graininess} of the quantum phase space. In a typical chaotic dynamics the motion in phase space $\Gamma$ is bounded, occupying the system a finite region $R$ of volume $\textrm{vol}(R)$. In turn, in the semiclassical limit $q=\frac{\textrm{vol}(R)}{\hbar^{D}}\gg1$ the value of $\textrm{vol}(R)$ can be approximated by the sum of all the rigid boxes $\Delta q\Delta p$ that are contained in $R$.
Let us call $c_1,c_2,\ldots,c_n$ to these boxes. In such situation the region $\Sigma$ corresponding to the rigid boxes that intersect the frontier of $R$ can be neglected. An illustration for $D=1$ is shown in Fig. 1.
Now, since there is no subset in grained phase space having a volume smaller than $\Delta q\Delta p$ it follows that
\begin{eqnarray}\label{generator6}
\widetilde{Q}=\{c_1,c_2,\ldots,c_n\}
\end{eqnarray}
is the unique generator of $R$, that we will call \emph{quantum generator}. Moreover, one has that any $\sigma$-algebra in quantum phase space can be only composed by unions of the rigid boxes $c_1,c_2,\ldots,c_n$.

\subsection{Estimation of the logarithmic timescale in the semiclassical limit}

In order to obtain a semiclassical version of Theorem \ref{estimation KS}, we consider a quantum system having a classical analogue $(\Gamma,\mu,\Sigma,\{T_t\})$ provided with a finely grained phase space $\Gamma$ in the semiclassical limit $q\gg1$, as is shown in Fig. 1.
%composed by rigid boxes $\Delta q\Delta p$
Also, the partition $\widetilde{Q}$ of \eqref{generator6} is the quantum generator of the region $R$ occupied by the classical analogue.
Let us assume that the dynamics in phase space is ergodic\footnote{Note that the condition of invertibility is guaranteed since the equations of motion in classical mechanics are time-reversal.}. Then, we can arrive to our main contribution of this work, established by means of the following result.
\begin{theorem}\label{estimation logarithmic}(Estimation of the time rescaled KS--entropy of the classical analogue)

\noindent Assume one has a quantum system having a classical analogue $(\Gamma, \Sigma, \mu, \{T_t\}_{t\in J})$ that occupies a region $R$ of a discretized quantum phase space of dimension $2D$. If $T=T_{\tau}$ is an ergodic and invertible measure--preserving transformation with $h_{KS}^{(\tau)}<\infty$ the time rescaled KS--entropy of the classical analogue, then in the semiclassical limit $q\gg1$ one has
\begin{eqnarray}\label{generator7}
e^{h_{KS}^{(\tau)}}\leq n \leq e^{h_{KS}^{(\tau)}}+1
\end{eqnarray}
where $n=\frac{\textrm{vol}(R)}{\hbar^{D}}$ is the quasiclassical parameter $q$ and $\widetilde{Q}=\{c_1,c_2,\ldots,c_n\}$ is the quantum generator.
% number of rigid boxes $\Delta q\Delta p$ contained in $R$ that coincides with the quasiclassical parameter.
\end{theorem}
\begin{proof}
It is clear that the partition $\widetilde{Q}=\{c_1,c_2,\ldots,c_n\}$ of eq. \eqref{generator6} is the only quantum generator in the quantum phase space and since $\textrm{vol}(R)$ is $n$ times the volume $\hbar^{D}$ of each rigid box contained in $R$ then one obtains $\textrm{vol}(R)=n\hbar^{D}$. Then,
the result follows by applying the Theorem \ref{estimation KS} to the classical analogue in the semiclassical limit.
\end{proof}
\noindent From the equation \eqref{generator7} one has
\begin{eqnarray}\label{generator9}
\tau h_{KS}\leq \log q \leq \log (e^{\tau h_{KS}}+1) \nonumber
\end{eqnarray}
from which follows that
\begin{eqnarray}\label{generator10}
\tau \leq \frac{\log q}{h_{KS}} \leq  \frac{\log (e^{\tau h_{KS}}+1)}{h_{KS}} \nonumber
\end{eqnarray}
Now, assuming a chaotic motion of the classical analogue by means of the condition $h_{KS}>0$, then one can make the approximation $e^{\tau h_{KS}}+1\approx e^{\tau h_{KS}}$. Replacing this in the last inequality one has
\begin{eqnarray}\label{generator11}
\tau = \frac{\log q}{h_{KS}}  \ \ \ \ \ \textrm{with} \ \ \ \ \ q=\frac{\textrm{vol}(R)}{\hbar^{D}}
\end{eqnarray}
which is precisely the logarithmic timescale.

\section{Physical relevance}

Here we provide a discussion about the physical relevance of the results obtained at the light of the quantum chaos theory.
Several previous work based on the quantum dynamics of observable values \cite{Ber78}, quantization by means of symmetric and ordered expansions \cite{Ang03}, and the wave--packet spreadings along hyperbolic trajectories \cite{Sch12} among others, show that a unified scenario for a characterization of the quantum chaos timescales is still absent.
Furthermore, the mathematical structure used in most of these approaches makes difficult to visualize intuitively the quantum and classical elements that are present, or even in some cases the results are restricted to special initial conditions \cite{Ang03}.
Nevertheless, we can mention the aspects of our contribution in agreement with some standard approaches used. Below we quote some results of the literature and discuss them from the point of view of the present work.

\begin{itemize}
  \item \emph{The timescale $\tau_{\hbar}$ may be one of the universal and fundamentals characteristic of quantum chaos accessible to experimental observation \cite{Ber78,Cas95}. In fact, the existence of $\tau_{\hbar}$
      \begin{eqnarray}\label{universal timescale}
      \tau_{\hbar}=C_1\log\left(\frac{C_2}{\hbar}\right)
      \end{eqnarray}
      where $C_{1,2}$ are constants has been observed and discussed in detail for some typical models of quantum chaos \cite{Ber92,Zas81}.}
      \vspace{0.1cm}

The relation \eqref{universal timescale} results as a mathematical consequence of Theorem \ref{estimation logarithmic} for any phase space of arbitrary dimension $2D$. In fact, from \eqref{generator11} one obtains $C_1=\frac{1}{h_{KS}}$, and $C_2=\textrm{vol(R)}$ the volume occupied by the system along the dynamics.
\vspace{0.1cm}

  \item \emph{Every classical structure with a phase--space
area smaller than Planck's constant $\hbar$ has no quantum correspondence. Only
the classical structures extending in phase space over scales larger than Planck's
constant are susceptible to emerge out of quantum--mechanical waves \cite{Gas14}.}
     \vspace{0.1cm}

From Theorem \ref{estimation logarithmic} one can see that the classical structure of KS--entropy estimation (Theorem \ref{estimation KS}) emerges in terms of the quasiclassical parameter $q$ in the semiclassical limit, as is expressed in \eqref{generator7}.
\vspace{0.1cm}

  \item \emph{If strong chaos occurs in the classical limit, then for a rather short time $\tau_{\hbar}$ the wave--packet spreads over the phase volume:
 $\Delta I=\hbar\exp(\lambda \tau_{\hbar})$ where $\lambda$ is the characteristic Lyapunov exponent. Therefore, for the time--scale $\tau_{\hbar}$, one has: $\tau_{\hbar}\sim \lambda^{-1}\ln (\Delta I/\hbar)=(\ln \kappa)/\lambda$, where $\kappa$ is of the order of the number of quanta of characteristic wave packet width \cite{Cas95}.}
      \vspace{0.1cm}

In fact, from \eqref{generator7} with $D=1$ it follows that for the case of the wave--packet one has: $h_{KS}^{(\tau)}=\tau_{\hbar}\lambda$, $h_{KS}=\lambda$ (Pesin theorem), and $\kappa=q=\frac{\textrm{vol}(R)}{\hbar}$. In this way, the number of quanta of the characteristic wave packet width is equal to the number $n$ of the members of the quantum generator given by \eqref{generator6}.

\end{itemize}
\begin{figure}[th] \label{fig2}
\begin{center}
\includegraphics[width=14cm]{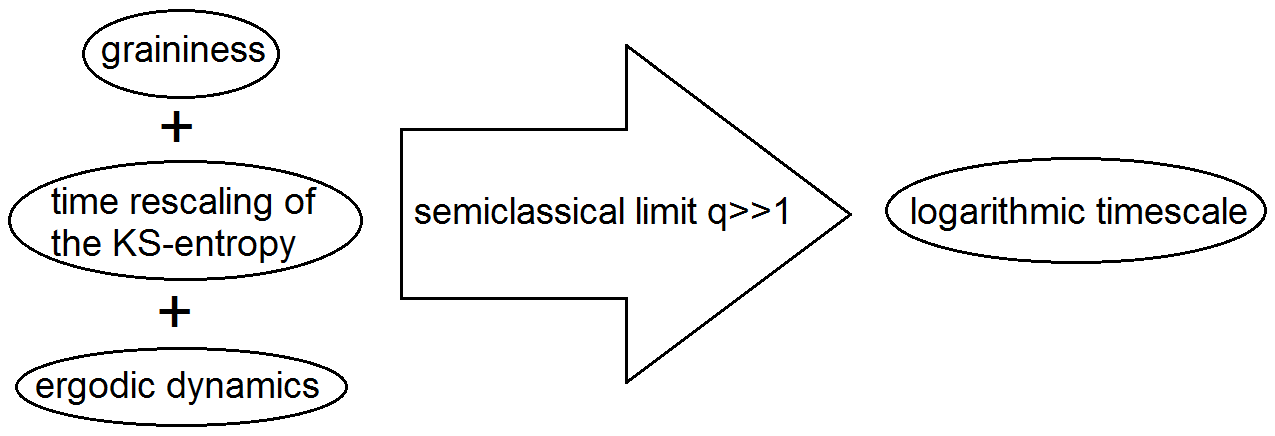}
\end{center}
\caption{An schematic picture of Theorem \ref{estimation logarithmic} showing the necessary elements for obtaining the logarithmic timescale.}
\end{figure}
A panoramic outlook of the content of Theorem \ref{estimation logarithmic} is shown in Fig. 2.

\section{Conclusions}

We have presented, in the semiclassical limit, an estimation of the logarithmic timescale for a quantum system having a classical analogue provided with an ergodic dynamics in its phase space. The three ingredients that we used were: 1) the fine granularity of the quantum phase space in the semiclassical limit, 2) the existence of an estimation of the KS--entropy in terms of finite generators of the region that the system occupies in phase space, and 3) a time rescaling of the KS--entropy that allows to introduce the characteristic timescale as a parameter.

In summary, our contribution is three--fold. On the one hand, the logarithmic timescale arises, in the semiclassical limit, as a formal result of the ergodic theory applied to a discretized quantum phase space of the classical analogue in an ergodic dynamics, thus providing a theoretical bridge between the ergodic theory and the graininess of the quantum phase space.
On the other hand, the Theorem \ref{estimation logarithmic} makes more visible and rigorous the simultaneous interaction of the effects of the quantum dynamics and the classical instability in phase space. In fact, the semiclassical parameter $q$ is expressed as the number of members of the quantum generator of the region that the classical analogue occupies in phase space.

Finally, one can consider the Theorem \ref{estimation logarithmic} as a mathematical proof of the existence of the logarithmic timescale when the dynamics of the classical analogue is chaotic, i.e. a positive KS--entropy. However, since the quasiclassical parameter $q$ and the KS--entropy are system--specific, in each example the parameters of the logarithmic timescale must be determined by the experimental observation.

One final remark. It is pertinent to point out that, in addition to Theorem \ref{estimation KS}, the techniques employed (i.e. the existence of a single quantum generator of the quantum phase space and the time rescaling property of the KS--entropy) can be used to extend semiclassically others results of the ergodic theory.

\section*{Acknowledgments}
This work was partially supported by CONICET and Universidad Nacional de La Plata, Argentina.

%\end{multicols}
\end{document}